\documentclass[journal]{IEEEtran}
\usepackage{tikz}
\usepackage{amsmath,amssymb,amsfonts}
\usepackage[ruled,vlined,linesnumbered]{algorithm2e}
\usepackage{algpseudocode}
\usepackage{xr}
\usepackage{pifont}
\usepackage{tabularx}
\usepackage{cellspace}
\makeatletter
\newcommand{\removelatexerror}
{\let\@latex@error\@gobble}
\makeatother
\newcommand{\smallcircled}[1]{{\textcircled{\fontsize{6pt}{7pt}\selectfont #1}}}
\usepackage{graphicx}
\usepackage{textcomp}
\usepackage{xcolor}
\usepackage{array}
\usepackage{textcomp}
\usepackage{dutchcal}
\usepackage{multirow}
\usepackage{bm}
\usepackage{booktabs}
\usepackage{subfigure} 
\usepackage{orcidlink} 
\usepackage{pifont}
\hypersetup{hidelinks}
\usepackage[mathscr]{euscript}
\usepackage{mathtools}
\usepackage{epstopdf}
\usepackage{cases}
\def\BibTeX{{\rm B\kern-.05em{\sc i\kern-.025em b}\kern-.08em
		T\kern-.1667em\lower.7ex\hbox{E}\kern-.125emX}}
\usepackage{cases}
\usepackage{amsthm}

\newtheorem{proposition}{\emph{Proposition}}

\newtheorem{remark}{\bf \emph{Remark}}

\begin{document}

\title{Low-Altitude UAV-Carried Movable Antenna for Joint Wireless Power Transfer and Covert Communications}

\author{Chuang~Zhang,\IEEEmembership{}
        Geng~Sun,~\IEEEmembership{Senior Member,~IEEE,}
        Jiahui~Li,\IEEEmembership{}
        Jiacheng~Wang,\IEEEmembership{}
        Qingqing~Wu,~\IEEEmembership{Senior Member,~IEEE,}
        Dusit~Niyato,~\IEEEmembership{Fellow,~IEEE,}
        Shiwen~Mao,~\IEEEmembership{Fellow,~IEEE,}
        and Tony~Q.~S.~Quek,~\IEEEmembership{Fellow,~IEEE}
        
        \thanks{Chuang Zhang is with the College of Computer Science and Technology, Jilin University, Changchun 130012, China, and also with the Singapore University of Technology and Design, Singapore 487372 (e-mail: chuangzhang1999@gmail.com).}
        \thanks{Geng Sun is with the College of Computer Science and Technology, Key Laboratory of Symbolic Computation and Knowledge Engineering of Ministry of Education, Jilin University, Changchun 130012, China, and also with the College of Computing and Data Science, Nanyang Technological University, Singapore 639798 (e-mail: sungeng@jlu.edu.cn).}
        \thanks{Jiahui Li is with the College of Computer Science and Technology, and also with the Key Laboratory of Symbolic Computation and Knowledge Engineering of Ministry of Education, Jilin University, Changchun 130012, China (e-mail: lijiahui0803@foxmail.com).}
        \thanks{Jiacheng Wang and Dusit Niyato are with the College of Computing and Data Science, Nanyang Technological University, Singapore 639798 (e-mails: jiacheng.wang@ntu.edu.sg, dniyato@ntu.edu.sg).}
        \thanks{Qingqing Wu is with the Department of Electronic Engineering, Shanghai Jiao Tong University, Shanghai, China (email: qingqingwu@sjtu.edu.cn).}
        \thanks{Shiwen Mao is with the Department of Electrical and Computer Engineering, Auburn University, Auburn, AL 36849, USA (e-mail: smao@ieee.org).}
        \thanks{Tony Q. S. Quek is with the Information System Technology and Design Pillar, Singapore University of Technology and Design, Singapore 487372 (e-mail: tonyquek@sutd.edu.sg).}
        \thanks{\textit{(Corresponding author: Geng Sun.)}}

}



\IEEEtitleabstractindextext{
\begin{abstract}	
The proliferation of Internet of Things (IoT) networks has created an urgent need for sustainable energy solutions, particularly for the battery-constrained spatially distributed IoT nodes. While low-altitude uncrewed aerial vehicles (UAVs) employed with wireless power transfer (WPT) capabilities offer a promising solution, the line-of-sight channels that facilitate efficient energy delivery also expose sensitive operational data to adversaries. This paper proposes a novel low-altitude UAV-carried movable antenna-enhanced transmission system joint WPT and covert communications, which simultaneously performs energy supplements to IoT nodes and establishes transmission links with a covert user by leveraging wireless energy signals as a natural cover. Then, we formulate a multi-objective optimization problem that jointly maximizes the total harvested energy of IoT nodes and sum achievable rate of the covert user, while minimizing the propulsion energy consumption of the low-altitude UAV. To address the non-convex and temporally coupled optimization problem, we propose a mixture-of-experts-augmented soft actor-critic (MoE-SAC) algorithm that employs a sparse Top-K gated mixture-of-shallow-experts architecture to represent multimodal policy distributions arising from the conflicting optimization objectives. We also incorporate an action projection module that explicitly enforces per-time-slot power budget constraints and antenna position constraints. Simulation results demonstrate that the proposed approach significantly outperforms some baseline approaches and other state-of-the-art deep reinforcement learning algorithms.
\end{abstract}

}

\maketitle
\IEEEdisplaynontitleabstractindextext
\IEEEpeerreviewmaketitle

%
%
\section{Introduction}
\label{sec_introduction}

\IEEEPARstart{T}{he} growth of Internet of Things (IoT) networks has transformed numerous industries and driven innovations in the fields like environmental monitoring, precision agriculture, disaster management and smart city development \cite{Pan2019}. Nevertheless, a large number of spatially distributed IoT nodes rely on batteries with limited energy capacity. As a result, large-scale maintenance operations for battery replacement become challenging and economically unsustainable, particularly in remote and hazardous deployment areas \cite{Pan2025}. To mitigate the inherent energy supply limitation, wireless power transfer (WPT) has arisen as a promising paradigm for enabling sustainable IoT operations through electromagnetic radiation \cite{Wang2016}. However, conventional WPT approaches remain constrained by limited transmission range and coverage area, thus rendering them insufficient for geographically dispersed IoT deployments. In such case, low-altitude uncrewed aerial vehicles (UAVs), owing to their maneuverability and capability of establishing line-of-sight (LoS) links, have been regarded as a promising platform for WPT in IoT networks \cite{Xie2025}. By optimizing flight trajectories from a starting point to a destination, the low-altitude UAVs can remotely deliver energy to spatially distributed IoT nodes as needed, thereby extending the lifespan of IoT networks.

\par During such energy replenishment missions, the low-altitude UAVs must maintain reliable communications with ground control stations or authorized users to report their operational status. Although the LoS channel condition enables efficient WPT, it also increases the risk of sensitive information leakage to potential adversaries. Traditional physical-layer security approaches aim to prevent successful decoding of transmissions by exploiting favorable channel conditions to maximize the secrecy rate, thereby ensuring that the intercepted signals remain unintelligible to unauthorized receivers \cite{Wang2019}. However, such approaches prove insufficient in security-critical scenarios such as police patrol, where preventing message decoding alone cannot prevent the revealing of operational information, as the mere detection of communication activity exposes mission presence and timing. In contrast, covert communications seek to render the optimal detector of an adversary statistically indistinguishable from random guessing by leveraging ambient noise or interference to mask the information signal \cite{Chen2023}. Notably, the energy signals emitted for WPT inherently serve as an active cover for embedding covert information, thereby forming a low-altitude UAV transmission system that integrates WPT and covert communications to simultaneously achieve sensitive information transmission and energy delivery for distributed IoT nodes.

\par Despite the advantages of such integration, the high propulsion energy consumption of low-altitude UAVs remains a major bottleneck since the battery capacity is limited by payload capability. Specifically, single-antenna configurations often require close-range operations to simultaneously enhance the WPT efficiency and communication covertness, thereby restricting the flight endurance of low-altitude UAVs. To address this limitation, low-altitude UAVs can carry antenna arrays to realize beamforming \cite{Xu2025}. Such configurations concentrate electromagnetic energy toward intended IoT nodes while suppressing signal leakage toward unintended directions. By spatially directing transmission beams, the low-altitude UAVs can achieve efficient WPT and covert communications from greater standoff distances, eliminating the need for close-proximity maneuvers and thereby substantially reducing propulsion energy consumption. Nevertheless, the beamforming performance of an antenna array is fundamentally determined by the spatial distribution of antenna elements. In light of this characteristic, movable antenna (MA) has recently emerged to provide enhanced flexibility \cite{Zhu2024}. Unlike traditional antenna arrays where the positions of antenna elements are permanently determined during manufacturing, MA enables real-time spatial reconfiguration by physically repositioning individual antenna elements within a designated movement region, thereby dynamically reshaping electromagnetic radiation patterns. 

\par However, designing and optimizing such a low-altitude UAV-carried MA-enhanced transmission system for joint WPT and covert communications presents several significant challenges. \textit{First}, such a system necessitates comprehensively considering the energy harvesting performance of ground IoT nodes, transmission performance of covert communication links, and propulsion energy consumption of the low-altitude UAV platform. \textit{Second}, inherent conflicts exist among these three aspects. For instance, allocating excessively high transmit power to energy beams can improve the harvested energy of IoT nodes, while it may reduce the achievable rate of the covert user. \textit{Finally}, optimization decisions at the current time instant directly influence future system states, thus requiring a balance between immediate effectiveness and long-term mission benefits. These challenges call for an intelligent sequential decision-making approach that can adaptively optimize multiple coupled optimization objectives while accounting for the temporal dependencies between current decisions and future system performance.

\par Accordingly, we introduce a novel online optimization approach to support the low-altitude
UAV-carried MA-enhanced transmission system for joint
WPT and covert communications. The main contributions of this work are outlined as follows:

\begin{itemize}
\item \textit{Low-altitude
UAV-Carried MA-Enhanced Transmission System for Joint WPT and Covert Communications:} We propose a low-altitude UAV-carried MA-enhanced transmission system joint WPT and covert communications that simultaneously delivers energy to IoT nodes and transmits sensitive information to the covert user. Specifically, the system utilizes WPT to serve as an active cover for covert communications by embedding the communication signal within the energy signal. Moreover, we integrate the MA into the low-altitude UAV to dynamically reconfigure the antenna array, thereby enhancing the energy transfer and decreasing information leakage. As far as we know, this is the first attempt to combine the MA with a low-altitude UAV to jointly support WPT and covert communications.

\item \textit{Multi-Objective Optimization Problem Formulation with Covert Requirements:} We analyze and derive the covert requirement of the considered system at each time slot, and formulate a multi-objective optimization problem that simultaneously maximizing the total harvested energy of all IoT nodes and sum achievable rate of the covert user, while minimizing the total propulsion energy consumption of the low-altitude UAV, by jointly optimizing the trajectory of the low-altitude UAV, precoding vectors and relative element positions of the low-altitude UAV-carried MA. This problem is non-convex with strong temporal coupling between decisions and constraints, thus making it particularly challenging to be solved efficiently.

\item \textit{Mixture-of-Experts (MoE)-Augmented Deep Reinforcement Learning (DRL) Algorithm:} We propose a mixture-of-experts-augmented soft actor-critic (MoE-SAC) algorithm within the framework of DRL to solve the formulated multi-objective optimization problem. Specifically, the MoE-SAC algorithm employs a sparse Top-K gated mixture-of-shallow-experts architecture to effectively represent the multimodal policy distributions, which enables the agent to better explore diverse strategies and balance conflicting objectives. Moreover, we incorporate an action projection module in MoE-SAC to explicitly enforce per-time-slot power budget and antenna position constraints, thereby ensuring the feasibility of the solutions and enhancing the convergence speed of the proposed algorithm.

\item \textit{Simulation and Performance Analysis:} We conduct extensive simulations to validate the proposed approach from both the system architecture and algorithmic perspectives, comparing its performance with several baselines. Simulation results validate the superiority of the proposed approach over existing methods. Moreover, the trajectory visualization results of the low-altitude UAV further explain how the algorithm balances the different optimization objectives and detours around the Warden to enhance the performance of the considered system.
\end{itemize}

\par The remainder of this paper is organized as follows. Section \ref{sec_related_work} reviews the related works, and Section \ref{sec_system_model} describes the system model. The covert requirement and the optimization problem are presented in Section \ref{sec_problem_formulation}. Section \ref{sec_moe_augmented_DRL_algorithm} introduces the proposed MoE-SAC algorithm, followed by simulation results in Section \ref{sec_performance_evaluation}. Finally, Section \ref{sec_conclusion} concludes the paper.

%
%
\section{Related Work}
\label{sec_related_work}

\par Existing research related to our work can be categorized into three main dimensions regarding low-altitude UAV-enabled WPT and covert communication system architectures, performance metrics in WPT and covert communications and optimization approaches.

\subsection{Low-altitude UAVs-enabled WPT and Covert Communication System Architectures}

\par Due to their agile maneuverability and operational flexibility, low-altitude UAVs have found extensive applications in both WPT and covert communications. In the context of WPT, Xu \textit{et al}. \cite{Xu2018} investigated that multiple UAVs equipped with energy transmitters cruise above the service region to cooperatively power spatially distributed ground energy receivers, thereby forming a multi-UAV air-to-ground WPT network. In \cite{Yuan2021}, the authors designed a low-altitude UAV platform equipped with a uniform linear phased array for analog beamforming, thus establishing a directional air-to-ground WPT link to distributed sensor nodes. In the realm of covert communications, Wang \textit{et al}. \cite{Wang2023} proposed a covert communication framework supported by a UAV-mounted intelligent reflecting surface, which dynamically shapes the wireless propagation environment to enable undetectable transmission despite the uncertain location of the Warden. Moreover, in \cite{Du2022}, the authors \textit{et al}. studied a jammer-aided covert communication system, where a multi-antenna UAV serves multiple ground users while a separate multi-antenna ground base station acts as a friendly jammer by injecting artificial interference to increase the detection uncertainty of the Warden. 

\par \textit{\textbf{Limitation 1:} All these architectures assume fixed antenna configurations and separately treat WPT and covert communications. These approaches fail to utilize WPT as an active cover for covert communications fully and do not exploit the spatial reconfigurability of antenna elements, thus fundamentally limiting the performance of the system.}

\subsection{Performance Metrics in WPT and Covert Communications}

\par Existing studies on WPT and covert communications are typically evaluated based on specific performance metrics aligned with their operational objectives. For example, in \cite{Dong2025}, the authors aimed to decrease the energy consumption of all low-altitude UAVs in large-scale IoT systems by jointly optimizing the number of UAVs and their trajectories. Moreover, Kim \textit{et al}. \cite{Kim2024} maximized the minimum uplink throughput of ground nodes by jointly optimizing the user scheduling, transmit power of ground nodes, and trajectories of UAVs,  subject to the energy constraint of the ground nodes. In \cite{Bhalerao2025}, the authors maximized the total harvested energy of sensor nodes by jointly optimizing the three-dimensional deployment of the aerial vehicle and energy allocation to spatially distributed nodes with heterogeneous energy demands. In terms of low-altitude UAVs for covert communications, Deng \textit{et al}. \cite{Deng2025} maximized the achievable covert rate for legitimate users by jointly optimizing the beamforming vectors and trajectories of the UAV under covert constraints dictated by multiple passive Wardens in an integrated sensing and communication network. Furthermore, in \cite{Huang2021}, the authors maximized the detection error probability of the Warden by jointly optimizing the transmit power and trajectory of the low-altitude UAV to establish the covert communication link. 

\par \textit{\textbf{Limitation 2:} Existing studies have predominantly addressed the harvested energy of IoT nodes, achievable rate of covert users and propulsion energy consumption of UAVs in isolation, thus lacking a unified optimization framework to balance these optimization objectives.}

\subsection{Optimization Approaches}

\par To solve the non-convex optimization problems in WPT and covert communications, current studies primarily rely on three categories of algorithms, i.e., convex approximation-based algorithms, heuristic algorithms and learning-based algorithms. For instance, in \cite{Hu2025}, the authors employed convex approximation techniques such as semidefinite relaxation and fractional programming within an alternating optimization framework to address the non-convex optimization problem in a UAV-enabled mobile edge computing system with simultaneous wireless information and power transfer. Moreover, Liu \textit{et al}. \cite{Liu2022} adopted a heuristic algorithm based on particle swarm for jointly optimizing the scheduling and trajectory design of the UAV in wireless rechargeable sensor networks. In \cite{Vo2025}, a deep neural network was trained to learn the relationship between environmental parameters and system configurations provided by offline optimization algorithms, thereby enabling rapid adaptation in dynamic UAV-assisted cognitive covert communication scenarios. More recently, Wang \textit{et al}. \cite{Wang2025a} leveraged a proximal policy optimization algorithm to jointly optimize the position and transmission parameters of the UAV in an uplink covert communication system where a covert IoT device coexists with multiple public IoT devices. 

\par \textit{\textbf{Limitation 3:} Although these algorithms demonstrate effectiveness in their respective scenarios, convex approximation-based and heuristic methods are computationally intensive and require frequent re-optimization for time-varying systems, whereas existing DRL algorithms suffer from limited policy expressiveness and fail to incorporate domain-specific structural constraints in continuous action spaces, thereby hindering their optimization performance.}

\par Different from existing studies, this paper addresses these limitations by considering a novel low-altitude UAV-carried MA-enhanced transmission system for joint WPT and covert communications, and developing an online optimization algorithm to augment the effectiveness of the considered system.

%
%
\section{System Model}
\label{sec_system_model}

\par In this section, we first introduce the considered low-altitude UAV-carried MA-enhanced transmission system for joint WPT and covert communications. Next, we provide a detailed description about the system model. Finally, we present the binary hypothesis testing procedure employed by the Warden.

\subsection{Scenario Description}

\begin{figure}[t]
	\centering
	\includegraphics[width=\linewidth]{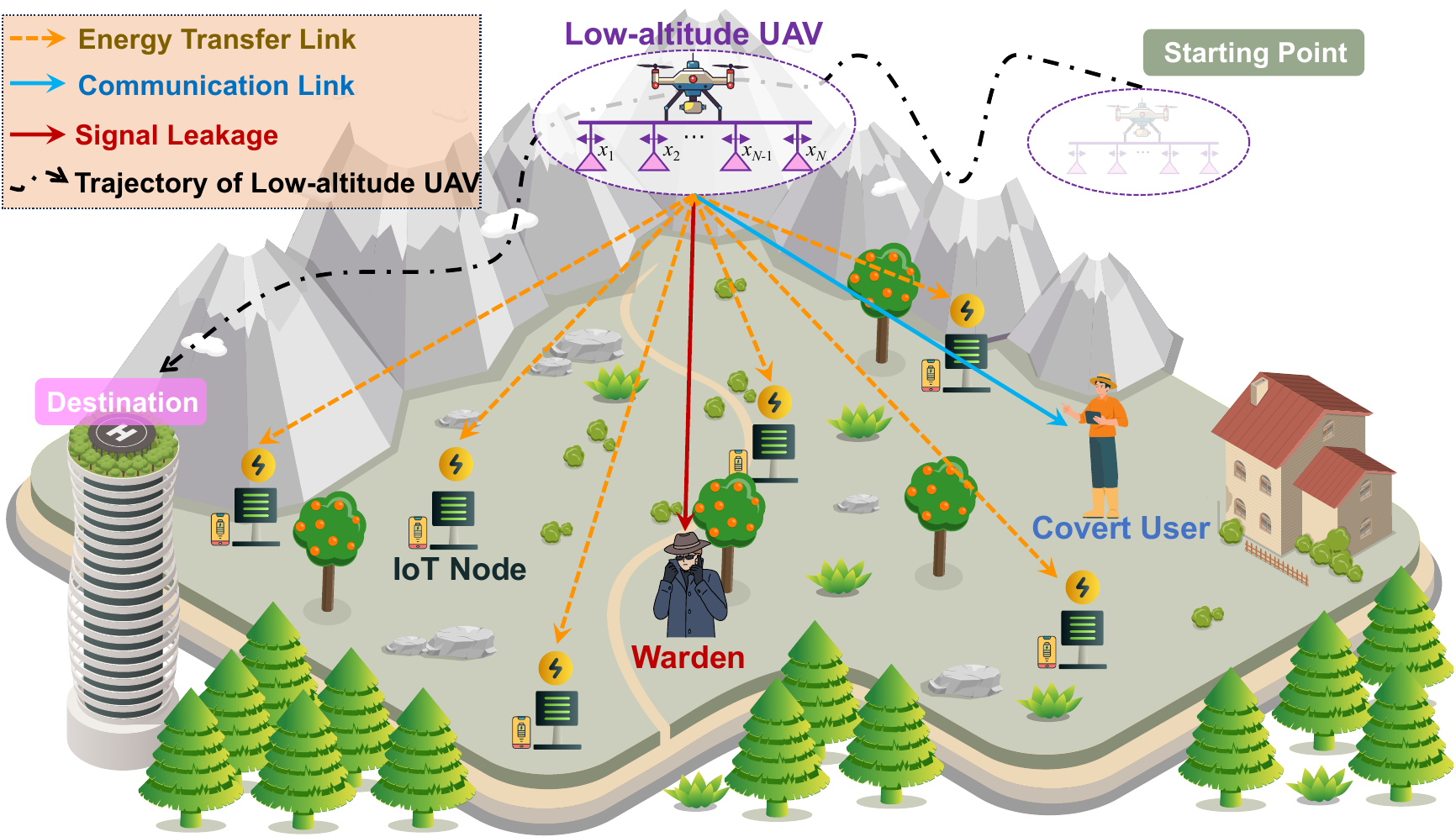}
	\caption{Illustration of the low-altitude UAV-carried MA-enhanced transmission system for joint WPT and covert communications, where the low-altitude UAV is equipped with an MA and performs WPT to IoT nodes to ensure their operational sustainability, while establishing a covert communication link with a covert user to report upon its status. Beyond powering the IoT nodes, the WPT also serves as an active cover to mask the covert transmission.}
	\label{Fig: System Model}
\end{figure}

\par As shown in Fig.~\ref{Fig: System Model}, we consider a low-altitude UAV-carried MA-enhanced transmission system for joint WPT and covert communications, where a low-altitude rotary-wing UAV $U$ is equipped with a linear MA array composed of $N$ antenna elements indexed by $\mathscr{N} = \{1, 2, \dots, N\}$. The low-altitude UAV simultaneously performs WPT to $M$ battery-constrained and low-power IoT nodes \cite{Xu2018}, indexed by $\mathscr{M} = \{1, 2, \dots, M\}$. Meanwhile, it establishes the covert communication links with a covert user $C$ to report upon its status or perceived environmental information in the presence of a ground Warden $W$ \cite{Huang2021}. In this setup, the battery-constrained IoT nodes rely on harvested radio-frequency energy for operational sustainability, while the covert user receives sensitive data under the cover provided by WPT, thereby making the Warden unable to distinguish the presence of covert communications.

\par For analytical tractability, the total mission duration $D$ is partitioned into $T$ consecutive and non-overlapping intervals of equal length $\delta_{t}$, yielding $D = T \delta_{t}$. Throughout the mission, the low-altitude UAV maintains a constant altitude $z_{\text{u}}$. As such, the resulting horizontal trajectory of the low-altitude UAV is captured by the coordinate sequence $\mathbf{q}_{\text{u}}(t) = [x_{\text{u}}(t), y_{\text{u}}(t), z_{\text{u}}]^{\text{T}}$, where the slot index $t$ ranges over $\mathscr{T} = \{1, \dots, T\}$. In accordance with practical deployment requirements \cite{Zhao2022}, the trajectory of the low-altitude UAV is constrained to originate at a fixed starting point $\mathbf{q}_{\text{u}}(1)=\mathbf{q}_{\text{s}}$ and terminate at a predetermined destination $\mathbf{q}_{\text{u}}(t)=\mathbf{q}_{\text{d}}$. Moreover, the $m$-th IoT node is located at $\mathbf{q}_{m} = (x_{m}, y_{m}, 0)$, while the covert user and Warden reside at $\mathbf{q}_{\text{c}} = [x_{\text{c}}, y_{\text{c}}, 0]^{\text{T}}$ and $\mathbf{q}_{\text{w}} = [x_{\text{w}}, y_{\text{w}}, 0]^{\text{T}}$, respectively.

\begin{remark}
This system architecture under consideration is particularly applicable to the mission-critical military scenarios. For example, the system supports energy-limited surveillance sensors via WPT while covertly transmitting battlefield intelligence in remote battlefields, thereby preventing hostile detection. Moreover, the low-altitude UAV is considered to operate at a fixed altitude and be capable of detecting the locations of the IoT nodes, covert user, and Warden via radar or camera, which is a widely adopted and realistic configuration in previous UAV trajectory optimization and covert communication studies \cite{Xie2025}, \cite{Xu2025}.
\end{remark}

\par In what follows, we introduce the comprehensive description for the channel model, communication model, WPT model, propulsion energy consumption model of the low-altitude UAV and binary hypothesis testing at the Warden.

\subsection{Channel Model}

\par Given the elevated operations of the low-altitude UAV and sparse ground obstacles, the air-to-ground channels can be viewed as LoS-dominated links \cite{Qian2025}. Accordingly, the path gain at time slot $t$ is characterized by using the free-space path loss channel, which can be expressed as follows:
\begin{equation}
    \label{Equ: path gain at time slot l}
    g_{i}(t) = \beta_{0}d_{\text{u},i}(t)^{-\alpha_{e}}, i \in \{m, \text{c}, \text{w}\}, 
\end{equation}

\noindent where $\beta_{0}$ represents the path loss at the reference distance of $1$ m, and $\alpha_{e}$ is the path loss exponent of air-to-ground channels. Moreover, $d_{\text{u},i}(t) = \| \mathbf{q}_{\text{u}}(t) - \mathbf{q}_{i} \|$ denotes the Euclidean distance between the low-altitude UAV and ground-level entity $i$ at time slot $t$. 

\par We consider the antenna elements of the low-altitude UAV-carried MA array arranged along the $x$-axis direction, and the position of the $n$-th antenna element relative to the low-altitude UAV is given by $x_{n}(t)$ at time slot $t$. Accordingly, the steering vector at time slot $t$ can be written as follows \cite{Liu2025}:
\begin{equation}
    \label{Equ: steering vector at time slot l}
    \begin{aligned}
        \mathbf{a}_{\text{u},i}(t) = \big[&e^{j\frac{2\pi}{\lambda}x_{1}(t)\cos\theta_{u,i}(t)}, e^{j\frac{2\pi}{\lambda}x_{1}(t)\cos\theta_{u,i}(t)}, \\
        &\dots, e^{j\frac{2\pi}{\lambda}x_{N}(t)\cos\theta_{u,i}(t)}\big]^{\text{T}}, i \in \{m, \text{c}, \text{w}\}
    \end{aligned}
\end{equation}

\noindent where $\theta_{\text{u},i} = \arcsin \left(\frac{z_{\text{u}}}{\|\textbf{q}_{\text{u}}(t)-\textbf{q}_{i}\|}\right)$ represents the steering angle towards the ground-level entity $i$, and $\lambda$ is the wavelength.

\par According to the descriptions above, the air-to-ground channel vector at time slot $t$ can be expressed as follows:

\begin{equation}
    \label{Equ: channel vector at time slot l}
    \mathbf{h}_{\text{u},i}(t) = \sqrt{g_{i}(t)} \, \mathbf{a}_{\text{u},i}(t), i \in \{m, \text{c}, \text{w}\}.
\end{equation}

\subsection{Communication and WPT Models}

\par Let $s_{\text{c}}(t) \sim \mathscr{CN}(0, 1)$ and $s_{\text{e}}(t)\sim \mathscr{CN}(0, 1)$ represent the normalized information signal and energy signal emitted by the low-altitude UAV at time slot $t$, respectively. Moreover, we denote $\mathbf{w}_{\text{c}}(t) \in \mathscr{C}^{N\times 1}$ and $\mathbf{w}_{\text{e}}(t) \in \mathscr{C}^{N\times 1}$ as precoding vectors for information and energy signals at time slot $t$, respectively. Therefore, the transmitted signal from the low-altitude UAV at time slot $t$ is described as follows:
\begin{equation}
    \label{Equ: transmitted signal}
    \textbf{s}_{\text{u}}(t) = \textbf{w}_{\text{c}}(t)s_{\text{c}}(t) + \textbf{w}_{\text{e}}(t)s_{\text{e}}(t).
\end{equation}

\par Consequently, the received signals of the covert user at time slot $t$ can be described as follows:
\begin{equation}
    \label{Equ: received signal of covert user}
    y_{\text{c}}(t) = \mathbf{h}_{\text{u},\text{c}}^{\text{H}}(t)\mathbf{s}_{\text{u}}(t) + n_{\text{c}}(t),
\end{equation}

\noindent where $n_{\text{c}}(t) \sim \mathscr{CN}(0, \sigma_{\text{c}}^{2})$ is the additive white Gaussian noise at the covert user. Thus, the achievable rate of the covert user at time slot $t$ can be calculated as follows:
\begin{equation}
    \label{Equ: achievable rate of covert user}
    R_{\text{c}}(t) = \log_{2}\left(1 + \frac{\left|\mathbf{h}_{\text{u},\text{c}}^{\text{H}}(t)\mathbf{w}_{\text{c}}(t)\right|^{2}}{\left|\mathbf{h}_{\text{u},\text{c}}^{\text{H}}(t)\mathbf{w}_{\text{e}}(t)\right|^{2}+\sigma_{c}^{2}}\right).
\end{equation}

\par Similarly, the received signals of the $m$-th energy-constrained IoT node can be written as follows:
\begin{equation}
    \label{Equ: received signal of IoT m}
    y_{m}(t) = \mathbf{h}_{\text{u},m}^{\text{H}}(t)\mathbf{s}_{\text{u}}(t) + n_{m}(t),
\end{equation}

\noindent where $n_{m}(t) \sim \mathscr{CN}(0, \sigma_{m}^{2})$ represents the additive white Gaussian noise at the $m$-th IoT node. As the noise power is much weaker than the signal power, it can be ignored in the energy harvesting process \cite{Zhou2024}. In practice, the energy receiver converts radio frequency power into direct-current power for charging the battery. Despite the non-linear nature of the conversion, the harvested direct-current power grows monotonically with the received radio frequency power. Thus, the linear energy harvesting model is adopted to characterize the harvesting process of the $m$-th IoT node at time slot $t$, which can be represented as follows:
\begin{equation}
    \label{Equ: linear EH model}
    E_{m}(t) = \eta \delta_{t}\left(|\mathbf{h}_{\text{u},m}^{\text{H}}(t)\textbf{w}_{\text{c}}(t)|^{2} + |\mathbf{h}_{\text{u},m}^{\text{H}}(t)\textbf{w}_{\text{e}}(t)|^{2}\right),
\end{equation}

\noindent where $\eta \in (0, 1]$  is the energy harvesting coefficient.

\subsection{Low-altitude UAV Propulsion Energy Consumption Model}

\par The propulsion energy consumption of the low-altitude rotary-wing UAV primarily depends on its flight dynamics. Following the widely adopted model in \cite{Zeng2017}, the propulsion energy consumption of the low-altitude UAV at time slot $t$ can be calculated as follows:
\begin{equation}
    \label{Equ: propulsion energy}
    \begin{aligned}
        E_{\text{u}}(t) = &\delta_{t}\Bigg(P_{i}\left(\sqrt{1+\frac{\|\mathbf{v}(t)\|^{4}}{4v_{0}^{4}}} -  \frac{\|\mathbf{v}(t)\|^{2}}{2v_{0}^{2}}\right)^{1/2} \\
        & P_{0}\left(1+\frac{3\|\mathbf{v}(t)\|^{2}}{U_{\text{tip}}^{2}}\right)
        + \frac{1}{2}d_{0}\rho s A \|\mathbf{v}(t)\|^{3}\Bigg),
    \end{aligned}
\end{equation}

\noindent where $P_{i}$ and $P_{0}$ denote the induced power and blade profile power in hovering, respectively, and $\mathbf{v}(t) = \frac{\mathbf{q}_{\text{u}}(k+1)-\mathbf{q}_{\text{u}}(t)}{\delta_{t}}$ represents the velocity of the low-altitude UAV at time slot $t$. Moreover, $U_{\text{tip}}$ is the tip speed of the rotor blade, $v_{0}$ is the mean rotor induced velocity in hovering, $d_{0}$ denotes the fuselage drag ratio, $\rho$ represents the air density, $s$ is the rotor solidity, and $A$ is the rotor disc area. 

\subsection{Binary Hypothesis Testing at Warden}

\par We denote the null hypothesis $\mathscr{H}_{0}$ and alternative hypothesis $\mathscr{H}_{1}$ as the two circumstances that only energy signals are transmitted and information signals alongside energy signals are transmitted, respectively. Specifically, the Warden determines whether the low-altitude UAV is transmitting information signals to the covert user by detecting differences in the received signal power under $\mathscr{H}_{0}$ and $\mathscr{H}_{1}$. Accordingly, the received signals of the Warden at time slot $t$ can be expressed as follows:
\begin{equation}
    \label{Equ: received signal power of the Warden}
    	y_{\text{w}}(t)=
    \begin{cases}
    	\begin{aligned}
    		&\mathbf{h}_{\text{w}}^{\text{H}}(t)\big(\textbf{w}_{\text{c}}(t)s_{\text{c}}(t) + \textbf{w}_{\text{e}}(t)s_{\text{e}}(t)\big) +  n_{\text{w}}(t),\\
    	\end{aligned} &  \mathcal{H}_{1},\\
    	\begin{aligned}
    		& \mathbf{h}_{\text{w}}^{\text{H}}(t) \textbf{w}_{\text{e}}(t)s_{\text{e}}(t) +  n_{\text{w}}(t),
    	\end{aligned} & \mathcal{H}_{0}, \\
	\end{cases}
\end{equation}

\noindent where $n_{\text{w}}(t)\sim \mathscr{CN}(0, \sigma_{\text{w}}^{2})$ is the additive white Gaussian noise at the Warden. Therefore, the average received power of the Warden at time slot $t$ can be described as follows:
\begin{equation}
    \label{Equ: average received power of the Warden}
    \mathbb{E}(|y_{\text{w}}(t)|^{2}) =
    \begin{cases}
    \begin{aligned}
    		&\mathbf{h}_{\text{w}}^{\text{H}}(t)\big(\textbf{W}_{\text{c}}(t)+ \textbf{W}_{\text{e}}(t)\big)\mathbf{h}_{\text{w}}(t) +  \sigma_{\text{w}}^{2},\\
    	\end{aligned} &  \mathcal{H}_{1},\\
    	\begin{aligned}
    		& \mathbf{h}_{\text{w}}^{\text{H}}(t) \textbf{W}_{\text{e}}(t)\mathbf{h}_{\text{w}}(t) + \sigma_{\text{w}}^{2},
    	\end{aligned} & \mathcal{H}_{0}, \\
    \end{cases}
\end{equation}

\noindent where $\textbf{W}_{\text{c}}(t) \triangleq \textbf{w}_{\text{c}}(t)\textbf{w}_{\text{c}}^{\text{H}}(t)$ and $\textbf{W}_{\text{e}}(t) \triangleq \textbf{w}_{\text{e}}(t)\textbf{w}_{\text{e}}^{\text{H}}(t)$.

\par In the binary testing problem, the decision rule of the Warden can be expressed as follows \cite{Wu2024}:
\begin{equation}
	\label{Equ: decision_rule}
\begin{aligned}
	\mathbb{E}(|y_{\text{w}}(t)|^{2}) \mathop{\gtrless}_{\mathcal{D}_{0}}^{\mathcal{D}_{1}}\tau(t),
\end{aligned}
\end{equation}
\noindent where $\tau(t)$ represents the decision threshold of the Warden at time slot $t$. Moreover, $\mathcal{D}_{0}$ and $\mathcal{D}_{1}$ are the decision results while supporting $\mathcal{H}_{0}$ and $\mathcal{H}_{1}$, respectively.

\par As such, we can observe that optimizing the trajectory of the low-altitude UAV, precoding vectors and antenna element positions of the low-altitude UAV-carried MA array can not only improve both the achievable rate of the covert user and harvesting energy of energy-constrained IoT nodes, but also mitigate the risk of detection by the Warden. 

%
%
\section{Problem Formulation}
\label{sec_problem_formulation}

\par In this section, we first derive the covert requirement of the considered low-altitude UAV-carried MA-enhanced transmission system for joint WPT and
covert communications at each time slot. Then, we formulate and analyze the optimization problem.

\subsection{Covert Requirement}

\par In this paper, we aim to investigate the worst-case situation where the Warden uses the optimal decision threshold to minimize the detection error probability.

\begin{proposition}
\label{lemma:1}
\par The optimal decision threshold for the Warden at time slot $t$ is determined as follows:
\begin{equation}
    \label{Equ: optimal decision threshold}
	\tau^{*}(t) = \zeta_{0}(t) \frac{1+\varrho(t)}{\varrho(t)}\ln\big(1+\varrho(t)\big),
\end{equation}

\noindent where $\zeta_{0}(t)= \mathbf{h}_{\mathrm{w}}^{\mathrm{H}}(t) \mathbf{W}_{\mathrm{e}}(t)\mathbf{h}_{\mathrm{w}}(t) + \sigma_{\mathrm{w}}^{2}$. Moreover, $\varrho(t) = \frac{\zeta_{1}(t) - \zeta_{0}(t)}{\zeta_{0}(t)}$ wherein $\zeta_{1}(t)$ is calculated by $ \mathbf{h}_{\mathrm{w}}^{\mathrm{H}}(t)\big(\mathbf{W}_{\mathrm{c}}(t)+ \mathbf{W}_{\mathrm{e}}(t)\big)\mathbf{h}_{\mathrm{w}}(t) +  \sigma_{\mathrm{w}}^{2}$.
\end{proposition}
\begin{proof}
Since $y_{\text{w}}(t)$ follows a Gaussian distribution under the two hypotheses, $p_{w}(t)$ follows exponential distributions with means $\zeta_{0}(t)$ and $\zeta_{1}(t)$, respectively. Therefore, the false alarm probability which declares declaring $\mathscr{H}_{1}$ when $\mathscr{H}_{0}$ is true can be expressed as follows:
\begin{equation}
    P_{\mathrm{FA}}(t) = \Pr(\mathscr{D}_{1}|\mathscr{H}_{0})
    = \int_{\tau(t)}^{\infty} \frac{1}{\zeta_{0}(t)} e^{-\frac{y}{\zeta_{0}(t)}} \mathrm{d}y
    = e^{-\frac{\tau(t)}{\zeta_{0}(t)}}.
\end{equation}

\par Similarly, the miss-detection probability which declares $\mathscr{H}_{0}$ when $\mathscr{H}_{1}$ is true can be respectively expressed as follows:
\begin{equation}
    P_{\mathrm{MD}}(t) = \Pr(\mathscr{D}_{0}|\mathscr{H}_{1})
    = \int_{0}^{\tau(t)} \frac{1}{\zeta_{1}(t)} e^{-\frac{y}{\zeta_{1}(t)}} \mathrm{d}y
    = 1 - e^{-\frac{\tau(t)}{\zeta_{1}(t)}}.
\end{equation}

\par Then, the total detection error probability of the Warden is described as follows:
\begin{equation}
    \label{Equ: P_FA+P_MD}
    \xi(t) = P_{\mathrm{FA}}(t) + P_{\mathrm{MD}}(t)
    = 1 + e^{-\frac{\tau(t)}{\zeta_{0}(t)}} - e^{-\frac{\tau(t)}{\zeta_{1}(t)}}.
\end{equation}

\par Note that $\xi(t)$ is convex regarding $\tau(t)$.The optimal threshold $\tau^{*}(t)$ minimizing $\xi(t)$ can be obtained by differentiating $\xi(t)$ with respect to $\tau(t)$ and setting the derivative equal to zero. Then, the optimal decision threshold in Eq. \eqref{Equ: optimal decision threshold} can be obtained.
\end{proof}

\begin{proposition}
\label{lemma:2}
The corresponding minimum detection error probability for the Warden is expressed as follows: 
\begin{equation}
    \label{Equ: minimal_DEP}
	\xi^{*}(t) = 1 + e^{-\frac{1+\varrho(t)}{\varrho(t)}\ln\big(1+\varrho(t)\big)}-e^{-\frac{1}{\varrho(t)}\ln\big(1+\varrho(t)\big)}.
\end{equation}
\end{proposition}

\begin{proof}
Based on Lemma \ref{lemma:1}, the minimum detection error probability for the Warden can be calculated by substituting Eq. \eqref{Equ: optimal decision threshold} into Eq. \eqref{Equ: P_FA+P_MD}.
\end{proof}

\par Following Lemma \ref{lemma:2}, we can impose a covert requirement $\xi^{*}(t) \geq 1 - \xi$ to guarantee that the covert transmission remains undetectable under the worst-case scenario where $\xi$ is a small value to represent the level of the required covert communications.

\subsection{Problem Formulation}

\par In this work, we aim to maximize the total harvested energy of all IoT nodes and sum achievable rate of the covert user, while simultaneously minimizing the total propulsion energy consumption of the low-altitude UAV.

\par To this end, we introduce four sets of decision variables, defined as follows. \textit{(i)} $\mathscr{V} = \{\textbf{v}(t)| t \in \mathscr{T}\}$, representing the velocity control of the low-altitude UAV over time slots. \textit{(ii)} $\mathscr{W}_{\text{c}} = \{\textbf{w}_{\text{c}}(t)| t \in \mathscr{T}\}$, describing the precoding vectors for covert information transmission over time slots. \textit{(iii)} $\mathscr{W}_{\text{e}} = \{\textbf{w}_{\text{e}}(t)| t \in \mathscr{T}\}$, specifying the precoding vectors for WPT over time slots. \textit{(iv)} $\mathscr{X} = \{x_{n}(t)| n \in \mathscr{N}, t \in \mathscr{T}\}$, indicating the relative positions of antenna elements relative to the low-altitude UAV over time slots.




\par According to the abovementioned objectives and definition of decision variables, the multi-objective optimization problem of the considered low-altitude UAV-carried MA-enhanced transmission system for joint WPT and covert communications can be formulated as follows:
\begin{subequations}
    \label{optimization_problem}
\begin{align}
    &\mathop{\operatorname{max}}\limits_{ \{\mathscr{V}, \mathscr{W}_{\text{c}}, \mathscr{W}_{\text{e}}, \mathscr{X}\}} F=\{\sum_{t=1}^{T}\sum_{m=1}^{M} E_{m}(t), \sum_{t=1}^{T} R_{\text{c}}(t), -\sum_{t=1}^{T} E_{\text{u}}(t)\}  \label{optimization_problem_objective}  \\
	&~~~~~~\textrm{s.t.}  \quad~~~\xi^{*}(t) \geq 1 - \xi, \forall t \in \mathscr{T}, \label{optimization_problem_constraint_1}\\ 
    &~~~~~~~~~~~~~~~ \Vert\mathbf{v}(t)\Vert \leq \text{v}_{\text{max}}, \forall t \in \mathscr{T}, \label{optimization_problem_constraint_2} \\
    &~~~~~~~~~~~~~~~ \text{x}_{\text{min}} \leq x_{\text{u}}(t) \leq \text{x}_{\text{max}}, \forall t \in \mathscr{T}, \label{optimization_problem_constraint_3} \\
    &~~~~~~~~~~~~~~~ \text{y}_{\text{min}} \leq y_{\text{u}}(t) \leq \text{y}_{\text{max}}, \forall t \in \mathscr{T}, \label{optimization_problem_constraint_4} \\
    &~~~~~~~~~~~~~~~ \mathbf{q}_{\text{u}}(1) = \mathbf{q}_{\text{s}},~~ \mathbf{q}_{\text{u}}(t) = \mathbf{q}_{\text{d}},\label{optimization_problem_constraint_5} \\
    &~~~~~~~~~~~~~~~ \Vert\mathbf{w}_{\text{c}}(t)\Vert^{2} + \Vert\mathbf{w}_{\text{e}}(t)\Vert^{2} \leq \text{p}_{max}, \forall t \in \mathscr{T}, \label{optimization_problem_constraint_6}\\
    &~~~~~~~~~~~~~~~ x_{n}(t) \in [0, L], \forall n \in \mathscr{N}, \forall t \in \mathscr{T}, \label{optimization_problem_constraint_7} \\
    &~~~~~~~~~~~~~~~ x_{i}(t) - x_{j}(t) \geq \frac{\lambda}{2}, \forall i,j \in \mathscr{N}, i\neq j, \forall t \in \mathscr{T},\label{optimization_problem_constraint_8}
\end{align} 
\end{subequations}
\noindent where constraint \eqref{optimization_problem_constraint_1} ensures that the covert requirement of the communication process is satisfied, while constraint \eqref{optimization_problem_constraint_2} limits the velocity of the low-altitude UAV to its maximum allowable velocity. Moreover, constraints \eqref{optimization_problem_constraint_3} and \eqref{optimization_problem_constraint_4} confine the low-altitude UAV within the predefined horizontal flight area, and constraint \eqref{optimization_problem_constraint_5} fixes its initial and final positions. In addition, constraint \eqref{optimization_problem_constraint_6} restricts the total transmit power of the combined covert information and energy signals to the maximum allowable value, and constraint \eqref{optimization_problem_constraint_7} ensures that the antenna elements of the low altitude UAV-carried MA array remain within the length of the antenna array. Furthermore, constraint \eqref{optimization_problem_constraint_8} guarantees the minimum spacing between antenna elements to avoid mutual coupling effects.

\par We analyze the characteristics of the above-formulated optimization problem as follows. 
\begin{itemize}
    \item \textit{Trade-offs among Optimization Objectives:} For example, increasing the steering energy beams toward IoT nodes may improve the harvested energy accumulated by IoT nodes, however, this also simultaneously reduces the achievable rate of the covert user. Similarly, shortening the trajectory of the low-altitude UAV to save propulsion energy may limit the harvesting energy of IoT nodes or degrade the achievable rate of the covert user.

    \item \textit{Non-convex of Optimization Problem:} Constraints \eqref{optimization_problem_constraint_1} and \eqref{optimization_problem_constraint_8} are also non-convex, which introduces additional difficulties in identifying feasible solutions.

    \item \textit{Long-term Optimization:} For the formulated optimization problem, the position of the low-altitude UAV at one time slot affects the feasible decisions and performance outcomes in subsequent slots. 
\end{itemize}

\par In such cases, the optimization problem is challenging to solve analytically by using traditional methods, i.e., convex optimization. Therefore, this motivates us to design a novel DRL-based method, which is well-suited for sequential decision-making problems characterized by high-dimensional decision spaces, non-convex constraints and temporally dependent dynamics.

\section{MoE-Augmented DRL-based Algorithm}
\label{sec_moe_augmented_DRL_algorithm}	

\par In this section, the considered optimization problem is reformulated as a Markov decision process (MDP). Then, we introduce the proposed MoE-augmented DRL-based algorithm. Finally, we give the complexity of the designed algorithm.

\subsection{MDP Reformulation}

\par In a typical DRL-based framework, a decision-maker referred to as the agent interacts with the external environment through the sequential selection and execution of actions. Formally, the sequential decision-making process is formulated as an MDP, which is defined by the 5-tuple $\Omega = <\mathscr{S}, \mathscr{A}, R, P, \gamma>$ \cite{Li2025}. At each time step $t$\footnote{To keep the symbols brief, the time slot index $t$ defined in parentheses in the system model will be uniformly represented by the subscript $t$ in the subsequent algorithm description section.}, the agent observes the current state $s_{t} \in \mathcal{S}$ of the external environment and responds by selecting an action $a_{t} \in \mathcal{A}$ based on its policy $\pi(a_{t} | s_{t})$. The environment then transforms to the next state $s_{t+1}$ according to the transition probability $P(s_{t+1} | s_{t}, a_{t})$ and returns a scalar reward $r_{t} = R(s_{t}, a_{t})$ to the agent. Moreover, the discount factor $\gamma \in (0,1]$ ensures that future rewards are appropriately weighted, thereby allowing the agent to pursue a policy that maximizes the long-term performance rather than myopic gains. In the following, we unfold how each element of the MDP is instantiated in our considered optimization problem.

\subsubsection{State Space}

\par The state space provides the agent with a comprehensive representation of the current circumstance of the environment. In the considered low-altitude UAV-carried MA-enhanced transmission system for joint WPT and covert communications, the state can be expressed as follows:
\begin{equation}
    \label{Equ: state space}
    \begin{aligned}
        s_{t} = \Big\{&t, \Re(\mathbf{H}_{t}), \Im(\mathbf{H}_{t}), \mathbf{d}_{t}^{\text{u,d}}, \mathbf{d}_{t}^{\text{u,c}}, \mathbf{d}_{t}^{\text{u,w}},\\
    & \mathbf{d}_{t}^{\text{u,} m}, \delta_{t-1}^{n} \Big| m \in \mathscr{M}, n \in \mathscr{N} \setminus \{N\}\Big\},
    \end{aligned}
\end{equation}

\noindent where $t$ denotes the current time slot index to indicate the progress of the mission and $\mathbf{H}_{t} = \{\mathbf{h}_{\text{u,c}}(t), \mathbf{h}_{\text{u},m}(t) | m \in \mathscr{M}\}$ represents the complex-valued channel matrix at time slot $t$. Since neural networks operate on real numbers, $\mathbf{H}_{t}$ is decomposed into its real part $\Re(\mathbf{H}_{t})$ and imaginary part $\Im(\mathbf{H}_{t})$, which are concatenated into a single real-valued vector as input. Moreover, $\mathbf{d}_{t}^{\text{u,d}}$, $\mathbf{d}_{t}^{\text{u,c}}$, $\mathbf{d}_{t}^{\text{u,w}}$ and $\mathbf{d}_{t}^{\text{u,}m}$ are distance vectors between the low-altitude UAV and its destination, covert user, Warden and IoT nodes alongside $x$-axis and $y$-axis, respectively. Furthermore, $\delta_{t-1}^{n}$ is the antenna spacing between the $n$-th and $(n+1)$-th movable antennas of the low-altitude UAV-carried MA array at time slot $t-1$.

\subsubsection{Action Space}

\par The action space reflects the set of controllable optimization variables that the agent can adjust at each time step to influence the system performance. As such, the action contains three parts that align with the decision variables related to the formulated optimization problem, which can be represented as follows:
\begin{equation}
     \label{Equ: action space}
     a_{t} = \big\{\mathbf{v}_{t}, \mathbf{w}_{t}^{\text{c}}, \mathbf{w}_{t}^{\text{e}}, x_{t}^{n}\big|n\in\mathscr{N}\big\},
\end{equation}

\noindent where $\mathbf{v}_{t}$ denotes the velocity of the low-altitude UAV at time slot $t$. Moreover, $\mathbf{w}_{t}^{\text{c}}$ and $\mathbf{w}_{t}^{\text{e}}$ represent the precoding vectors for covert communication and WPT, respectively, at time slot $t$, and $x_{t}^{n}$ represents the relative positions of the antenna elements at time slot $t$.

\subsubsection{Reward Function}

\par The reward signal is utilized to provide feedback that guides the policy toward achieving the optimal objective. Thus, we include the total harvested energy of all IoT nodes, achievable rate of the covert user and propulsion energy consumption at each time slot, which can be expressed as follows:
\begin{equation}
    r_{t}^{1} = w_{1} \sum_{m=1}^{M} E_{m}(t) + w_{2} R_{c}(t) - w_{3}E_{u}(t),
\end{equation}

\noindent where $w_{1}$, $w_{2}$ and $w_{3}$ are the weights that balance the trade-offs among three objectives. In addition, to handle the flight boundary constraints of the low-altitude UAV and the covert communication constraint, we introduce penalty terms $p_{1}$ and $p_{2}$ to discourage infeasible actions, respectively. Thus, the constraint penalty can be represented as follows:
\begin{equation}
    \begin{aligned}
    r_{t}^{2} =  & p_{1}\cdot\mathbb{I}\big(x_{\text{u}}(t) \notin [\text{x}_{\text{min}}, \text{x}_{\text{max}}] \lor y_{\text{u}}(t) \notin [\text{y}_{\text{min}}, \text{y}_{\text{max}}]\big) \\
    &+ p_{2}\cdot\mathbb{I}(\xi^{*}(t) < 1-\xi),
    \end{aligned}
\end{equation}

\noindent where $\mathbb{I}(\cdot)$ equals $1$ when its argument is true, and $0$ otherwise. Moreover, we incorporate a sparse terminal reward at time slot $T$ for the low-altitude UAV achieving the designated target position at the end of the mission, which can be written as follows:
\begin{equation}
    r_{T}^{3} = \begin{cases}
    \begin{aligned}
    		& r_{d},\\
    	\end{aligned} &  \mathbf{q}_{\text{u}}(T) = \mathbf{q}_{\text{d}},\\
    	\begin{aligned}
    		& -r_{d},
    	\end{aligned} & \mathbf{q}_{\text{u}}(T) \neq \mathbf{q}_{\text{d}}, \\
    \end{cases}
\end{equation}

\noindent where $r_{d}>0$ is a positive constant denoting the bonus for successfully reaching the destination. However, relying solely on such a sparse terminal reward may lead to inefficient exploration and slow convergence since the agent only receives meaningful feedback at the end of the episode. Thus, we adopt the reward shaping to address this challenge by providing additional intermediate signals that guide the low-altitude UAV toward the desired destination more effectively, which can be represented as follows:
\begin{equation}
    r_{t}^{4} = w_{4} \left(\left\Vert\mathbf{d}_{t-1}^{\text{u,d}}\right\Vert-\left\Vert\mathbf{d}_{t}^{\text{u,d}}\right\Vert\right),
\end{equation}

\noindent where $w_{4}>0$ is a shaping coefficient that balances this term against other reward components. 

\par Accordingly, the reward signal at time slot $t$ can be expressed as follows:
\begin{equation}
    r_{t} = r_{t}^{1} + r_{t}^{2} + r_{T}^{3} + r_{t}^{4}.
\end{equation}

\begin{remark}
Penalty terms are commonly employed to ensure the feasibility of solutions. However, they often lead to sparse rewards and increased variance \cite{Sutton1998}. Therefore, we apply penalty terms only to the complex constraints \eqref{optimization_problem_constraint_1}, \eqref{optimization_problem_constraint_3}, \eqref{optimization_problem_constraint_4} and \eqref{optimization_problem_constraint_5}, while the simpler constraints \eqref{optimization_problem_constraint_6}, \eqref{optimization_problem_constraint_7} and \eqref{optimization_problem_constraint_8} are directly enforced through the action projection module, which will be described in detail later.
\end{remark}
\par With the reformulated MDP, the original optimization problem is converted into a sequential decision-making task that can be addressed by the DRL-based framework. In the following, we design the specific algorithm to solve the sequential decision-making task related to our considered low-altitude UAV-carried MA-enhanced transmission system for joint WPT and covert
communications.

\subsection{Standard Soft Actor-critic (SAC) Algorithm}

\par SAC algorithm is one of the most widely adopted DRL algorithms for continuous control tasks \cite{Haarnoja2018}. Owing to its maximum entropy reinforcement learning framework, SAC offers superior robustness, sample efficiency and training stability compared to other DRL algorithms. By adding an entropy term into the objective, SAC is designed to encourage exploration and avoid early convergence to suboptimal deterministic policies. Formally, SAC aims to find an optimal policy $\pi_{\theta}^*$ that can maximize both the expected return and entropy of the policy, which is expressed as follows:
\begin{equation}
    \pi^{*} = \arg\max_{\pi}\sum_{t=0}^{T}\mathbb{E}
    \Big[ r(s_{t},a_{t}) + \alpha \mathscr{H}(\pi(\cdot|s_{t})) \Big],
\end{equation}
\noindent where $\mathscr{H}(\pi(\cdot|s_{t}))=-\mathbb{E}_{a_{t}\sim\pi}[ \log \pi(a_{t}|s_{t}) ]$ denotes the policy entropy at state $s_{t}$, and $\alpha>0$ is the temperature parameter that balances reward maximization and policy exploration.

\par To accomplish the above objective, SAC adopts an actor–critic framework that incorporates two soft Q-networks $Q_{\theta_{1}}(s, a)$ and $Q_{\theta_{2}}(s, a)$, their corresponding target networks $Q_{\bar{\theta}_{1}}(s, a)$ and $Q_{\bar{\theta}_{2}}(s, a)$, and an actor network $\pi_{\phi}(a|s)$ parameterized by a Gaussian distribution. Accordingly, the soft Q-networks are optimized by minimizing the Bellman residual loss, which can be given as follows:
\begin{equation}
    \label{Equ: soft-q-network update}
    \mathscr{L}_{Q}(\theta_{i}) = \mathbb{E}
    \Bigg[ \frac{1}{2}\Big(Q_{\theta_{i}}(s_{t},a_{t}) - y_{t}\Big)^{2}\Bigg], i \in \{1,2\},
\end{equation}
\noindent where $y_{t}$ is the target value, defined as follows:
\begin{equation}
    \begin{aligned}
    y_{t} = r_{t} + \gamma \mathbb{E}\big[ &\min_{i=1,2} Q_{\bar{\theta}_{i}}(s_{t+1},a_{t+1}) - \\
    & \alpha \log \pi_{\phi}(a_{t+1}|s_{t+1}) \big].
    \end{aligned}
\end{equation}

\par Derived from the maximum entropy principle \cite{Ziebart2008}, the policy is updated by minimizing the objective as follows:
\begin{equation}
    \mathscr{L}_{\pi}(\phi) = \mathbb{E}
    \Big[ \alpha \log \pi_{\phi}(a_{t}|s_{t}) - \min_{i=1,2} Q_{\theta_{i}}(s_{t},a_{t}) \Big].
\end{equation}

\par To adaptively control the balance between exploitation and exploration, the temperature parameter $\alpha$ is learned by minimizing the objective as follows:
\begin{equation}
    \label{Equ: alpha update}
    \mathscr{L}(\alpha) = \mathbb{E} \Big[ -\alpha \big(\log \pi_{\phi}(a_{t}|s_{t}) + \bar{\mathscr{H}}\big) \Big],
\end{equation}
\noindent where $\bar{\mathscr{H}}$ is a target entropy that specifies the desired exploration level.

\par However, the standard SAC algorithm still faces two primary challenges when addressing the formulated optimization problem. On the one hand, despite employing the entropy maximization to encourage exploration, SAC parameterizes its policy as a unimodal Gaussian distribution, which inherently limits its ability to represent multiple distinct optimal action modes that may coexist in our considered environments. Consequently, the policy may converge to suboptimal averaged behaviors rather than capturing the true multimodal structure of the policy space. On the other hand, the unstructured action representation of SAC ignores critical intra-action semantics, i.e., precoding vectors have no per-dimension physical bounds and are subject to global norm constraints, while relative positions of antenna elements of the low-altitude UAV-carried MA array must respect spatial ordering. Neither of these structural properties can be captured by modeling action dimensions as independent Gaussian samples.

\subsection{MoE-SAC Algorithm}

\begin{figure*}[t]
	\centering
	\includegraphics[width=\linewidth]{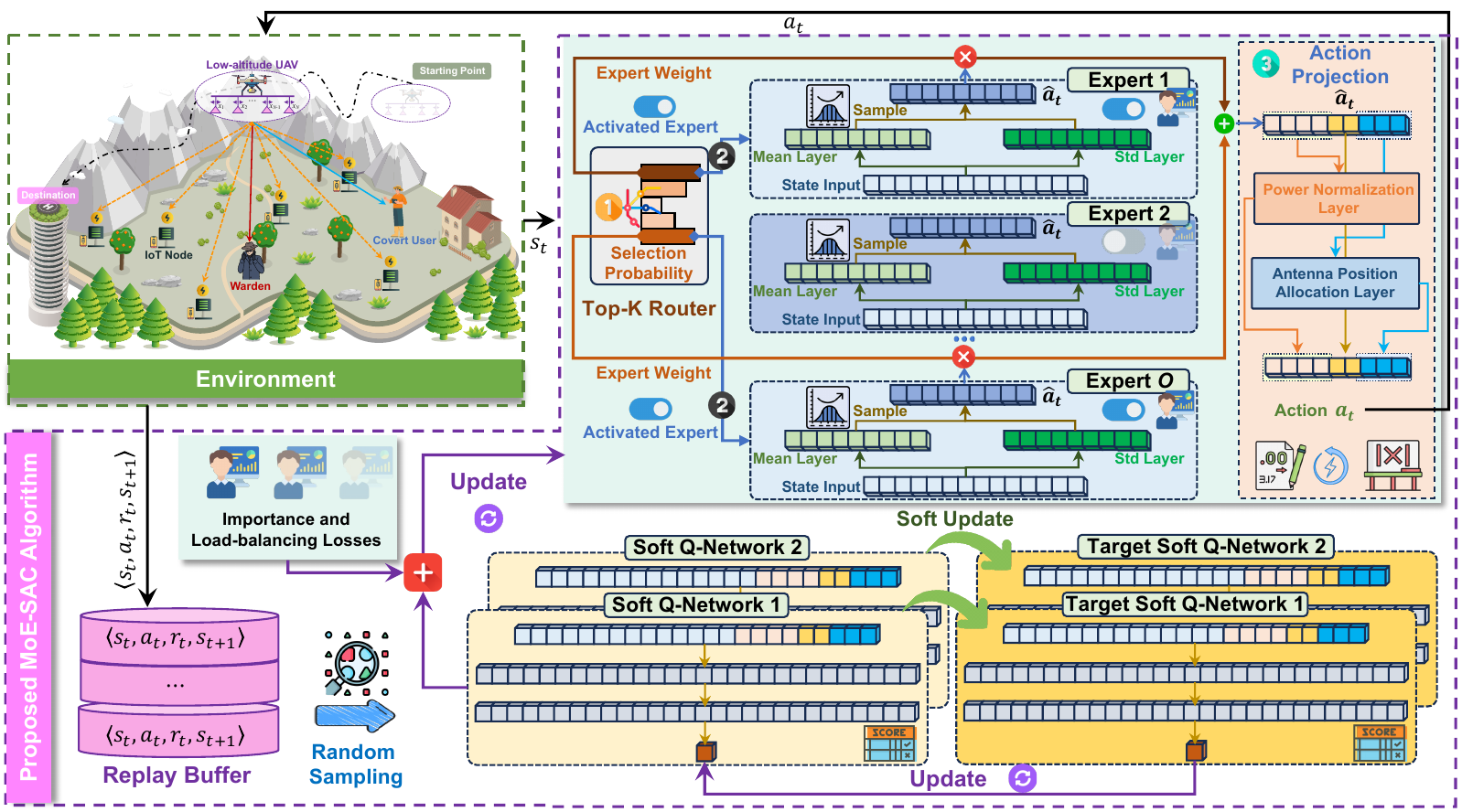}
	\caption{Framework of the MoE-SAC algorithm for the considered low-altitude UAV-carried MA-enhanced transmission system for joint WPT and covert communications. \smallcircled{1} The MoE-SAC algorithm employs a sparse Top-K router to select the K experts with the highest expert weights. \smallcircled{2} A mixture-of-shallow-experts architecture composed of the selected K experts is adopted to jointly model multimodal policy distributions. \smallcircled{3} The sampled raw actions are projected through an action projection module to explicitly enforce per-time-slot power budget constraints and antenna position constraints.}
	\label{Fig: Algorithm Framework}
\end{figure*}
\par To address the drawbacks of standard SAC for the formulated optimization problem, we propose a novel MoE-SAC algorithm that introduces two critical enhancements. First, we replace the conventional unimodal policy with an MoE-enhanced actor network, thereby enabling the representation of multimodal action distributions that capture multiple distinct optimal strategies. Then, we introduce an action projection module that maps raw policy outputs onto the feasible action space, which ensures that the generated precoding vectors and relative positions of antenna elements respect physical constraints and intra-action structural semantics. Accordingly, Fig.~\ref{Fig: Algorithm Framework} depicts the overall framework of the proposed MoE-SAC algorithm, and the following subsections detail the design of the MoE-enhanced actor network and action projection module.

\subsubsection{MoE-enhanced Actor Network}

\par To overcome the unimodality limitation while decreasing the additional overhead brought by MoE as much as possible, we adopt a sparse mixture-of-shallow-experts architecture for the actor network. Specifically, the MoE-enhanced actor network consists of two core components, i.e., a trainable router and multiple lightweight shallow policy networks.

\par Formally, the router takes the current state of the environment as input and produces a set of logits by a parameterized router $\hat{g}_{\Upsilon}(s_{t})$, each representing the expert weight of one expert. To achieve this, we employ the Top-K routing mechanism where the router selects only the $K$ experts with the highest expert weights. This process can be mathematically expressed as follows:
\begin{equation}
    \label{Equ: routing network}
    g_{\Upsilon}(s_{t}) = \mathrm{Softmax}\left(\mathrm{ToP}_{\mathrm{K}}\left(\mathrm{Softmax}\left(\hat{g}_{\Upsilon}(s_{t})\right)\right)\right),
\end{equation}

\noindent where $\mathrm{Softmax}(\cdot)$ refers to the softmax activation function, and $\mathrm{ToP}_{\mathrm{K}}(\cdot)$ is an operator that preserves only the $K$ largest entries while setting all others to $0$.

\par Then, we calculate the final output of the MoE-enhanced actor network as a weighted combination of multiple single-layer shallow experts, which can be expressed as follows:
\begin{equation}
    \label{Equ: actor network}
    \pi_{\phi}(\hat{a}_{t}|s_{t}) = \sum_{o=1}^{O} \left[g_{\Upsilon}(s_{t})\right]_{o} \cdot \, \pi_{\phi_{o}}(\hat{a}_{t}|s_{t}),
\end{equation}

\noindent where $O$ is the total number of experts. Moreover, $\left[g_{\Upsilon}(s_{t})\right]_{o}$ and $\pi_{\phi_{{o}}}(a_{t}|s_{t})$ denote the corresponding expert weight and action distribution produced by the $o$-th expert given state $s_{t}$, respectively.

\par Inspired by \cite{Riquelme2021}, we further augment the objective function of the MoE-enhanced actor network with an auxiliary loss that jointly encourages balanced expert utilization, which is defined as follows:
\begin{equation}
    \begin{aligned}
        \mathscr{L}_{\mathrm{aux}} &(\Upsilon) =
\frac{1}{2} \left( 
    \frac{
        \mathrm{Std}\left( 
            \left\{ \sum_{s_{t}} \left[ g_{\Upsilon}(s_{t}) \right]_o \right\}_{o=1}^O 
        \right)
    }{
        \mathrm{Mean}\left( 
            \left\{ \sum_{s_{t}} \left[ g_{\Upsilon}(s_{t}) \right]_o \right\}_{o=1}^O 
        \right)
    } 
\right)^2 \\+ &
\frac{1}{2} \left( 
    \frac{
        \mathrm{Std}\left( 
            \left\{ \sum_{s_{t}} \left[ \mathrm{Softmax}\left( \hat{g}_{\Upsilon}(s_{t}) \right) \right]_o \right\}_{o=1}^O 
        \right)
    }{
        \mathrm{Mean}\left( 
            \left\{ \sum_{s_{t}} \left[ \mathrm{Softmax}\left( \hat{g}_{\Upsilon}(s_{t}) \right)  \right]_o \right\}_{o=1}^O 
        \right)
    } 
\right)^2,
    \end{aligned}
\end{equation}

\noindent where $\mathrm{Mean}(\cdot)$ and $\mathrm{Std}(\cdot)$ denote the mean and standard deviation operators, respectively. Moreover, the first item and the second term of the auxiliary loss represent the importance and load-balancing losses, respectively. Thus, the final loss of the MoE-enhanced actor network can be written as follows:
\begin{equation}
    \label{Equ: actor network update}
\begin{aligned}
    \mathscr{L}_{\pi}(\phi) = \mathscr{L}_{\mathrm{aux}}(\Upsilon) + \mathscr{L}_{\pi}(\phi).
\end{aligned}
\end{equation}

\subsubsection{Action Projection Module}

\par The raw actions sampled from the MoE-enhanced actor network do not inherently respect constraints \eqref{optimization_problem_constraint_6}, \eqref{optimization_problem_constraint_7} and \eqref{optimization_problem_constraint_8}. To address this, we introduce an action projection module that maps raw policy output to the feasible action space.

\par For the precoding vector component of the low-altitude UAV-carried MA array, the raw action includes the precoding vectors for information and energy signals. To satisfy the per-time-slot power constraint, the power normalization layer is introduced to obtain the legal precoding vectors, which can be represented as follows:
\begin{equation}
    \label{Equ: power normalization layer}
    \hat{\mathbf{w}_{t}^{i}} =
    \begin{cases}
    \sqrt{\text{p}_{\max}}\cdot \tanh(\mathbf{w}_{t}^{i}), & P_{t}  \leq \text{p}_{\max}, \\[0.8ex]
    \dfrac{\sqrt{\text{p}_{\max}}\cdot \tanh(\mathbf{w}_{t}^{i})}{P_{t}}, & \text{otherwise},
\end{cases}
\quad i \in \{\text{c}, \text{e}\},
\end{equation}

\noindent where $P_{t} = \Vert \mathbf{w}_{t}^{\text{c}}\Vert ^{2} + \Vert \mathbf{w}_{t}^{\text{e}}\Vert ^{2}$ represents the total transmit power induced by the raw actions at time slot $t$ and $\tanh(\cdot)$ ensures that the network outputs remain bounded.

\par For the relative positions of antenna elements, the raw action is first decomposed into a total spacing variable $\varpi_{t}$ and a set of spacing ratios $\{\upsilon_{t}^{1}, \dots, \upsilon_{t}^{N-1}\}$. To ensure feasibility, the normalized total span and spacing ratios are given as follows:
\begin{equation}
\begin{cases}
    \hat{\varpi}_{t} = \dfrac{\tanh(\varpi_{t})+1}{2},\\[0.8ex]
    \hat{\upsilon}_{t}^{n} = \dfrac{\exp(\upsilon_{t}^{n})}{\sum_{j=1}^{N-1}\exp(\upsilon_{t}^{j})}, \quad n \in \mathscr{N}\setminus\{N\}.
\end{cases}
\end{equation}

\par Then, the inter-antenna spacing is determined as follows:
\begin{equation}
    \delta_{t}^{n} = \hat{\upsilon}_{t}^{n} \cdot \hat{\varpi}_{t} \cdot \left(L - (N-1)\frac{\lambda}{2}\right) + \frac{\lambda}{2}, \quad n \in \mathscr{N} \setminus \{N\},
\end{equation}
\noindent where $\delta_{t}^{n}$ denotes the spacing between the $n$-th and $(n+1)$-th movable antennas. Since the performance is determined by inter-antenna spacings rather than absolute positions \cite{Kang2024}, the antenna positions can be recursively obtained as follows:
\begin{equation}
    \label{Equ: ma positions layer}
\begin{cases}
    x_{t}^{1} = 0, \\
    x_{t}^{n} = x_{t}^{n-1} + \delta_{t}^{n-1}, \quad n \in \mathscr{N} \setminus \{1\}.
\end{cases}
\end{equation}

\begin{figure}[t]
	\removelatexerror
	\begin{algorithm}[H]
		\raggedright
		\caption{MoE-SAC Algorithm}
		\label{Algorithm 1: MoE-SAC Algorithm}
        \SetKwComment{tcc}{\textcolor{purple}{$\unrhd$ }}{}
        \LinesNumbered
        \tcc{Training Phase \textcolor{purple}{$\unlhd$} }
        Initialize the MoE-enhanced actor network parameters $\phi$, and soft Q-network parameters $\theta_{1}$ and $\theta_{2}$ as well as their corresponding target network parameters $\bar\theta_{1}$ and $\bar\theta_{2}$\;
        $\bar\theta_{1} \gets \theta_{1}, \bar\theta_{2} \gets \theta_{2}$ \;
        Initialize replay buffer $\mathscr{B} = \varnothing$\;
        \For{each episode}{
            Reset environment and initialize state $s_{0}$\;  
            \For{each time step $t$}{
                Select action $\hat{a}_{t} \sim \pi_{\phi}(a_{t}|s_{t})$\;
                Execute action projection to obtain $a_{t}$ according to Eqs. \eqref{Equ: power normalization layer} and \eqref{Equ: ma positions layer}\;
                Execute action $a_{t}$ in environment, then observe reward $r_{t}$ and next state $s_{t+1}$\;
                Store transition $(s_{t}, a_{t}, r_{t}, s_{t+1})$ in replay buffer $\mathscr{B}$\;
                
                \If{current\_step $\geq$ start\_learning\_step}{
                    Sample a batch from $\mathscr{B}$\;
                    Compute the loss of soft Q-networks by Eq. \eqref{Equ: soft-q-network update} and update soft Q-networks by the gradient descent method\;
                    Compute the loss of the MoE-enhanced actor network by Eq. \eqref{Equ: actor network update} and update the MoE-enhanced actor network by the gradient descent method\;
                    Update temperature parameter $\alpha$ according to Eq. \eqref{Equ: alpha update}\;
                    Update target soft Q-networks by the soft update mechanism $\bar{\theta}_{i} \gets \tau\theta_{i} + (1-\tau)\bar{\theta}_{i},~i = \{1, 2\};$
                }
            }
        }
        \setcounter{AlgoLine}{0}
        \tcc{Execution Phase \textcolor{purple}{$\unlhd$} }
        \For{each time step $t$}{
            Select action $\hat{a}_{t} \sim \pi_{\phi}(a_{t}|s_{t})$\;
            Execute action projection to obtain $a_{t}$ according to Eqs. \eqref{Equ: power normalization layer} and \eqref{Equ: ma positions layer}\;
        }
	\end{algorithm}
\end{figure}

\subsection{Main Steps of MoE-SAC Algorithm}

\par The main steps of the MoE-SAC algorithm can be summarized in Algorithm \ref{Algorithm 1: MoE-SAC Algorithm}, and consists of two sequential phases, i.e., the training phase and the execution phase. 

\par In the training phase, the proposed MoE-SAC algorithm begins with a warm-up period to populate the replay buffer with initial exploration data. During this period, the low-altitude UAV as agent interacts with the environment by sampling actions from the randomly initialized policy network, applying projection operations to enforce feasibility constraints, executing these actions to observe state transitions and rewards, and storing the resulting tuples in the replay buffer until sufficient samples are collected to enable stable off-policy learning. Then, the algorithm enters the training period to perform off-policy updates by sampling a mini-batch of transitions from the replay buffer. The twin soft Q-networks are first updated via Eq. \eqref{Equ: soft-q-network update} to minimize the temporal difference error, thereby learning accurate value estimates. Subsequently, the policy network parameters are updated via Eq. \eqref{Equ: actor network update} through policy gradient ascent to maximize the expected cumulative reward. Concurrently, the entropy temperature parameter is automatically adjusted via Eq. \eqref{Equ: alpha update} to dynamically balance exploration and exploitation throughout the training process. Finally, the target Q-networks are soft-updated using exponential moving averages with a momentum coefficient $\tau$, which stabilizes the learning dynamics by providing slowly-changing targets.

\par In the execution phase after training convergence, the MoE-SAC algorithm only needs to employ the trained MoE-enhanced actor network and action projection module to generate feasible actions, including the velocity of the low-altitude UAV, precoding vectors and the antenna element positions of the low-altitude UAV-carried MA array.

\subsection{Complexity Analyses}
\par The computational and space complexities of the proposed MoE-SAC algorithm are analyzed as follows.
\par \textbf{Training Phase}: The computational complexity of the MoE-SAC algorithm can be expressed as $\mathscr{O}(2|\phi| + 6|\theta| + ETG(|\Upsilon| + K|\phi_{o}|+Z) + BTG(2|\theta| + |\phi|))$, which can be broken down as follows \cite{Zhang2025}:

\begin{itemize}
    \item \textit{Network Initialization}: The computational complexity for initializing network parameters is $\mathscr{O}(|\phi| + 4|\theta|)$, where $|\phi|$ denotes the total number of parameters in the MoE-enhanced actor network, which includes router parameters $|\Upsilon|$ and all expert networks with the number of parameters $O\cdot|\phi_{o}|$, and $|\theta|$ represents the number of parameters in the soft Q-network.

    \item \textit{Action Generation and Execution}: The complexity of action generation and execution is $\mathscr{O}(ETG(|\Upsilon| + K|\phi_{o}|+Z))$, where $E$ is the total number of training episodes, $T$ represents the number of steps per episode, $G$ is the complexity of environment interactions, and $Z$ is the complexity of the action projection module.

    \item \textit{Network Updates}: This phase consists of soft Q-network updates and MoE-enhanced actor network updates. The soft Q-network updates based on Eq. \eqref{Equ: soft-q-network update} have complexity $\mathscr{O}(2BTG|\theta|)$, while updating the actor network based on Eq. \eqref{Equ: actor network update} has complexity $\mathscr{O}(BTG|\phi|)$, where $K$ is the number of experts activated per step via the ToP-$K$ routing mechanism. Moreover, the complexity of soft update is $\mathscr{O}(2|\theta| + |\phi|)$
\end{itemize}

\par The space complexity of the MoE-SAC algorithm accounts for the neural network parameters and replay buffer. This can be expressed as $\mathscr{O}(|\phi| + 4|\theta| + |D|(2|s| + |a| + 1))$, where $|D|$ is the replay buffer size, and $|s|$ and $|a|$ represent the dimensions of state and action spaces.

\par \textbf{Execution Phase}: During execution, the computational complexity of the MoE-SAC algorithm is $\mathscr{O}(T(|\Upsilon| + K|\phi_{o}| + Z))$, which includes the forward pass of the MoE-enhanced actor network with router computation and $K$ activated experts, plus the action projection module. Accordingly, the space complexity is $O(|\phi|)$, primarily from storing the parameters of the trained MoE-enhanced actor network.

\section{Simulation Results and Analysis}
\label{sec_performance_evaluation}

\par In this section, we present the simulation results and performance analysis of the proposed low-altitude UAV-carried MA-enhanced transmission system for joint WPT and covert communications.

\subsection{Simulation Setups}

\subsubsection{Scenario Description and Algorithm Setup}

\par We consider a low-altitude UAV-carried MA-enhanced transmission scenario in a $300~\text{m} \times 300~\text{m}$ area. Specifically, a low-altitude UAV equipped with $5$ antenna elements operates at a fixed altitude of $60$~m and simultaneously serves $4$ battery-constrained IoT nodes and one covert user. The total mission duration is discretized into $40$ time slots, each of duration $1$ s. The low-altitude UAV must fly from $[0, 0, 60]^\text{T}$ to $[280, 280, 60]^\text{T}$, with a maximum speed of $35$~m/s. Moreover, the required covertness level is set to $\xi = 0.1$ \cite{Wu2024}. The proposed MoE-SAC algorithm employs $8$ experts with $3$ activated experts. Furthermore, the soft Q-networks consist of two hidden layers with $256$ neurons, and the learning rates of all networks are set to $3 \times 10^{-4}$. In addition, other simulation parameters are shown in Table~\ref{Tab: Simulation Parameters}. 

\begin{table}[t]
\centering
\renewcommand{\arraystretch}{1.2}
\caption{Simulation Parameters}
\label{Tab: Simulation Parameters}
\begin{tabular}{ll|ll}
\toprule[1.5pt]
Parameter & Value & Parameter & Value \\
\hline
$L$ & $1.0$ m & $\lambda$ & $0.125$ m \\
$\alpha_{m}$ & $2$ & $\alpha_{\text{c}}$ & $2$\\ $\alpha_{\text{w}}$ & $2$& $\beta_0$ & $-40$ dB \\
$\sigma^2$ & $-140$ dBm  & $\eta$ & $0.6$\\
$|D|$ & $10^6$ & $B$ & $256$\\
$\gamma$ & $0.99$ & $\tau$ & $0.005$ \\
\bottomrule[1.5pt]
\end{tabular}
\end{table}

\subsubsection{Baselines}
\par To evaluate the performance of the proposed approach, we compare it against a set of representative baseline approaches, which are listed as follows:

\begin{itemize}
    \item \textit{Random}: At each time slot, the velocity of the low-altitude UAV, precoding vectors and antenna positions of the low-altitude UAV-carried MA array are randomly selected within the feasible physical space.
    
    \item \textit{Fixed Antenna Position (FAP)}: The relative positions of antenna elements are fixed at equally spaced intervals, while the trajectory of the low-altitude UAV and precoding vectors of the low-altitude UAV-carried MA array are optimized via the proposed MoE-SAC algorithm.

    \item \textit{Fixed Trajectory (FT)}: The low-altitude UAV follows a pre-determined straight-line path from the start point to the destination, while the precoding vectors and antenna element positions of the low-altitude UAV-carried MA array are jointly optimized via the proposed MoE-SAC algorithm.
\end{itemize}

\par In addition to the abovementioned system-level baselines, we also compare the proposed MoE-SAC algorithm with some other state-of-the-art DRL algorithms, which are listed as follows:
\begin{itemize}
    \item \textit{Deep Deterministic Policy Gradient (DDPG)}~\cite{Lillicrap2016}: A deterministic policy gradient algorithm that utilizes a single critic and actor network with target networks and replay buffer, serving as a representative DRL baseline.

    \item \textit{Twin Delayed Deep Deterministic Policy Gradient (TD3)}~\cite{Fujimoto2018}: An improved variant of DDPG that addresses overestimation bias via twin critics, delayed policy updates and target policy smoothing, which is known for its stability in continuous action spaces.

    \item \textit{Standard SAC}~\cite{Haarnoja2018}: The SAC with a unimodal Gaussian policy, automatic entropy tuning and twin soft Q-networks, which is considered as the state-of-the-art algorithm in maximum-entropy DRL for continuous tasks.
\end{itemize}
\begin{figure}[t]
	\centering
	\includegraphics[width=0.95\linewidth]{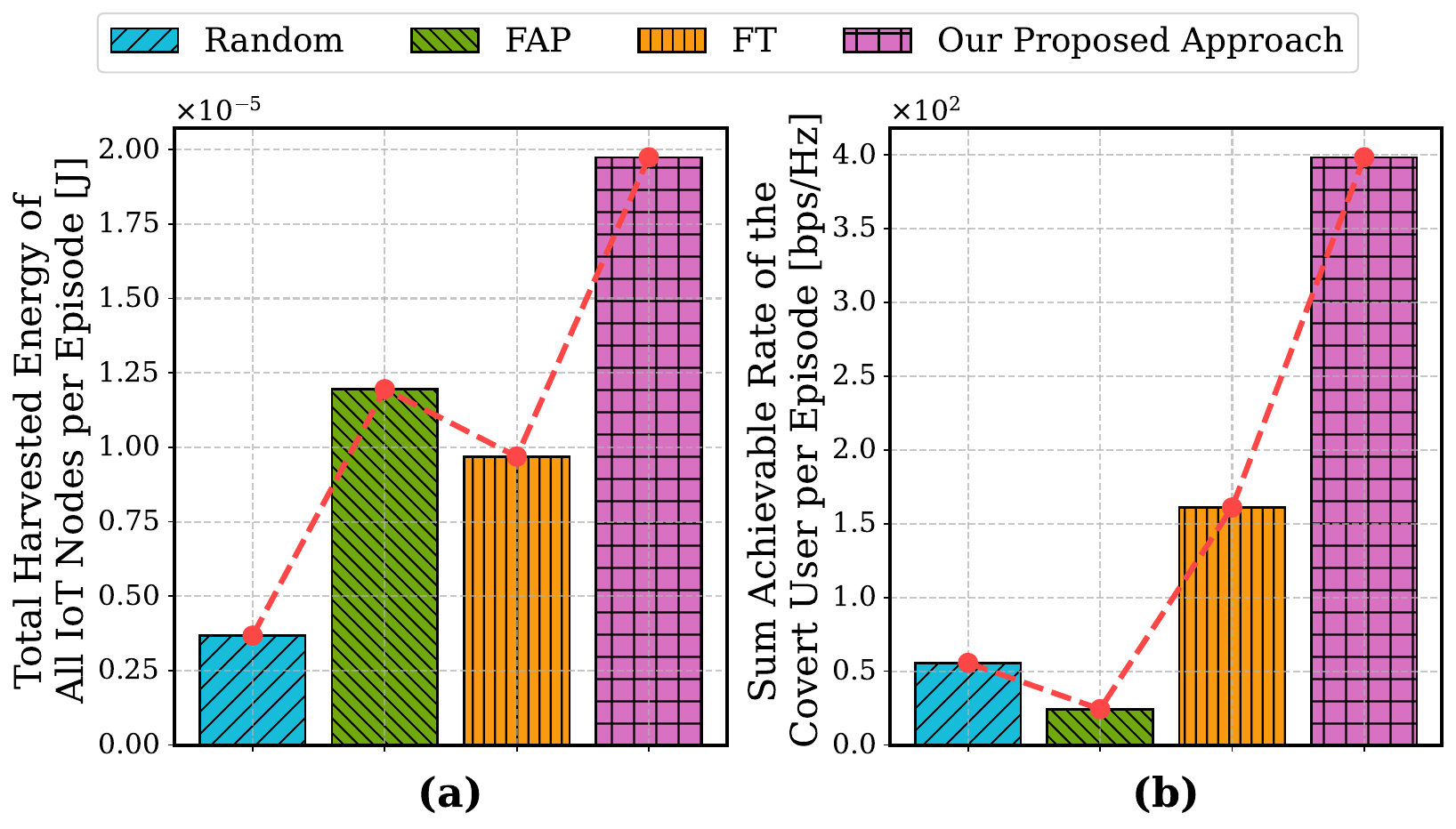}
	\caption{Performance comparison from different approaches: (a) Total harvested energy of all IoT nodes per episode and (b) Sum achievable rate of the covert user per episode.}
	\label{Fig: Performance comparison obtained by different approaches}
\end{figure}

\subsection{Performance Evaluation}
\begin{figure*}[t]
	\centering
	\includegraphics[width=\linewidth]{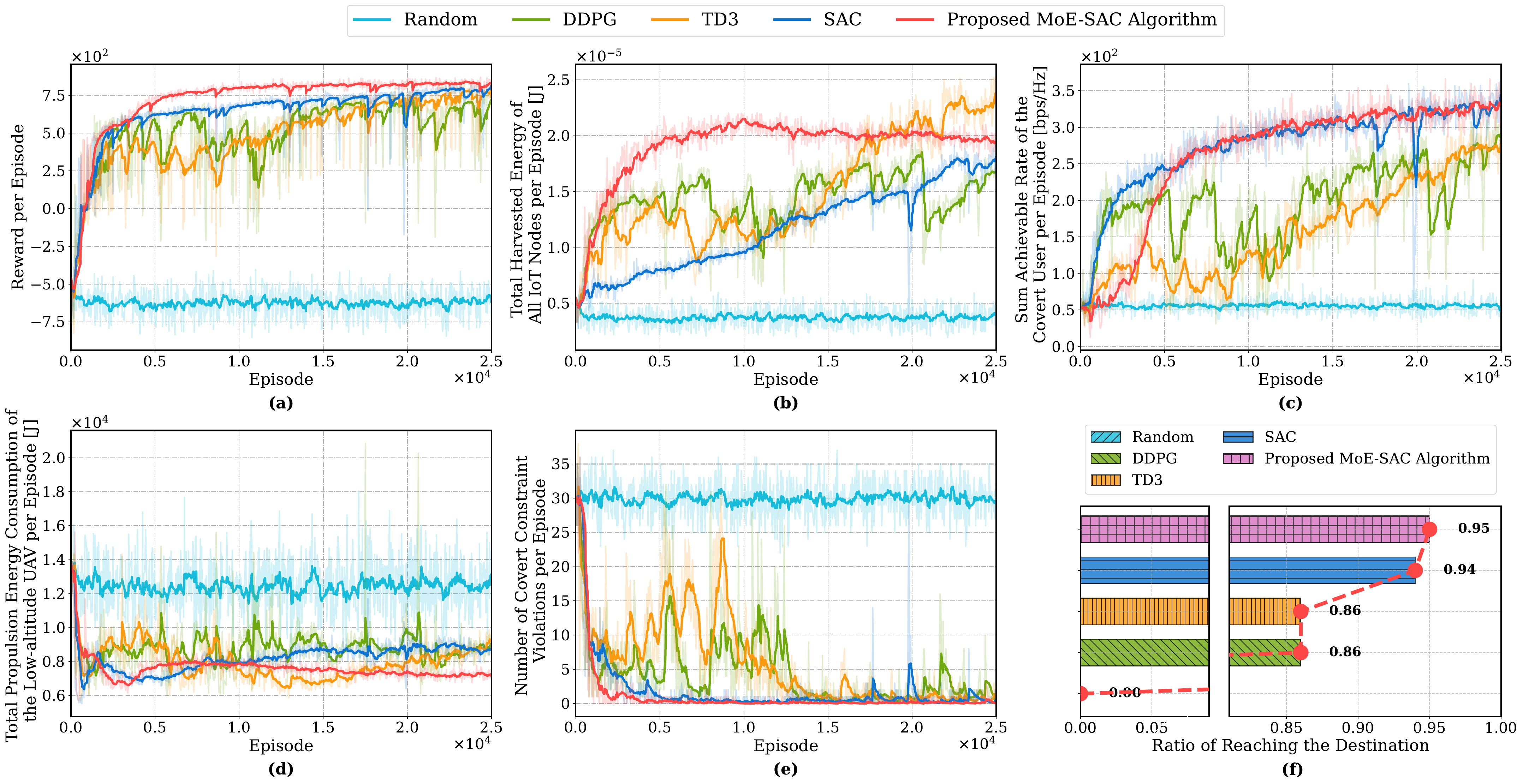}
	\caption{Training performance comparison of the proposed MoE-SAC algorithm against baseline algorithms: (a) Episode reward, (b) Total harvested energy of IoT nodes per episode, (c) Sum achievable rate of the covert user per episode, (d) Total propulsion energy consumption of the low-altitude UAV per episode, (e) Number of covert constraint violations per episode, and (f) Ratio of reaching the destination during the training.}
	\label{Fig: Training performance comparison}
\end{figure*}

\subsubsection{Comparison with Other Approaches}

\par Fig.~\ref{Fig: Performance comparison obtained by different approaches} presents the performance of our proposed approach in comparison with three system-level baselines. As shown in Figs.~\ref {Fig: Performance comparison obtained by different approaches}(a) and \ref {Fig: Performance comparison obtained by different approaches}(b), our proposed approach achieves the highest total harvested energy of all IoT nodes per episode and sum achievable rate of the covert user per episode, significantly outperforming FT and FAP. Specifically, this dual superiority stems from the joint optimization of the trajectory of the low-altitude UAV, precoding vectors and antenna element positions of the low-altitude UAV-carried MA array. By dynamically adjusting its flight path, the low-altitude UAV can approach both the IoT nodes and covert user to improve channel conditions, while the low-altitude UAV-carried MA array enables adaptive beamforming that simultaneously focuses energy toward the IoT nodes and forms a high-gain directional link toward the covert user to satisfy the covert constraint. In contrast, FT adopts a fixed straight-line trajectory, which limits its ability to approach ground IoT nodes and thus degrades both the energy harvesting of IoT nodes and achievable rate of the covert user. Although optimizing the trajectory of the low-altitude UAV and precoding vectors of the low-altitude UAV carried MA array, FAP fixes the element positions of the antenna array, thereby restricting beamforming flexibility and preventing effective simultaneous service to both the IoT nodes and covert user. Notably, Random achieves a slightly higher achievable rate for the covert user than FAP since it does not enforce the covert constraint during action selection. As a result, it may occasionally transmit with higher information power or favorable antenna configurations that strengthen the covert link, albeit at the expense of an increased detection risk.

\subsubsection{Comparison with Other DRL Algorithms}

\par Fig.~\ref{Fig: Training performance comparison} presents the training performance comparison of the proposed MoE-SAC algorithm against several representative DRL-based algorithms with Random serving as the worst-case performance. As can be seen in Fig.~\ref{Fig: Training performance comparison}(a), the reward curves of all algorithms exhibit significant oscillations during the initial training phase. This phenomenon arises from the insufficient experience of the agent in learning the action constraints, thus leading to frequent boundary and covertness violations that incur penalties. As the number of training episodes grows, all methods show convergence in their reward curves, reflecting progressively stabilized policy learning. Among all the compared algorithms, our proposed MoE-SAC algorithm attains a higher reward and exhibits faster convergence than the other methods. This is because the MoE-enhanced actor network enables the policy to capture multiple distinct optimal action modes, thereby facilitating faster exploration and better adaptation to our considered optimization problem.

\par Figs.~\ref{Fig: Training performance comparison}(b), \ref{Fig: Training performance comparison}(c), and \ref{Fig: Training performance comparison}(d) collectively reflect the optimization objectives of the system. We can observe that the proposed MoE-SAC algorithm not only achieves superior final performance but also demonstrates greater stability throughout the training process. Specifically, the learning curve of the MoE-SAC algorithm exhibits smoother convergence than DDPG and TD3, which suffer from frequent constraint violations due to their deterministic or less robust exploration mechanisms. Moreover, although benefiting from entropy-driven exploration, standard SAC is still limited by its unimodal Gaussian policy, which cannot represent the complex and multimodal nature of the optimal policy. In addition, Figs.~\ref{Fig: Training performance comparison}(e) and \ref{Fig: Training performance comparison}(f) illustrate the ability to satisfy critical system constraints under different algorithms. It can be seen that the proposed MoE-SAC algorithm exhibits the fewest violations of the covertness requirement and achieves a near-perfect success rate in reaching the destination, reflecting its superior capability to handle multiple coupled constraints simultaneously. This advantage primarily stems from the MoE-enhanced policy representation, which enables the agent to learn a diverse set of action strategies through multiple experts that balance the performance among communication, covertness and navigation.

\subsubsection{Low-altitude UAV Trajectory Visualization}

\par Fig.~\ref{Fig: Trajectory} visualizes the flight trajectory generated by different algorithms. Specifically, the proposed MoE-SAC algorithm clearly demonstrates a well-planned strategy that balances multiple objectives while ensuring mission completion. Along the trajectory, the low-altitude UAV takes a careful detour around the Warden near the bottom-right corner, showing its learned ability to avoid detection. While the trajectories generated by DDPG and TD3 generally follow a similar path, they do not approach the covert user as closely as the trajectory generated by the proposed MoE-SAC algorithm. In contrast, SAC guides the low-altitude UAV toward a nearby IoT node close to the Warden for energy transfer, inadvertently compromising covertness due to the reduced distance.
\begin{figure}[t]
	\centering
	\includegraphics[width=\linewidth]{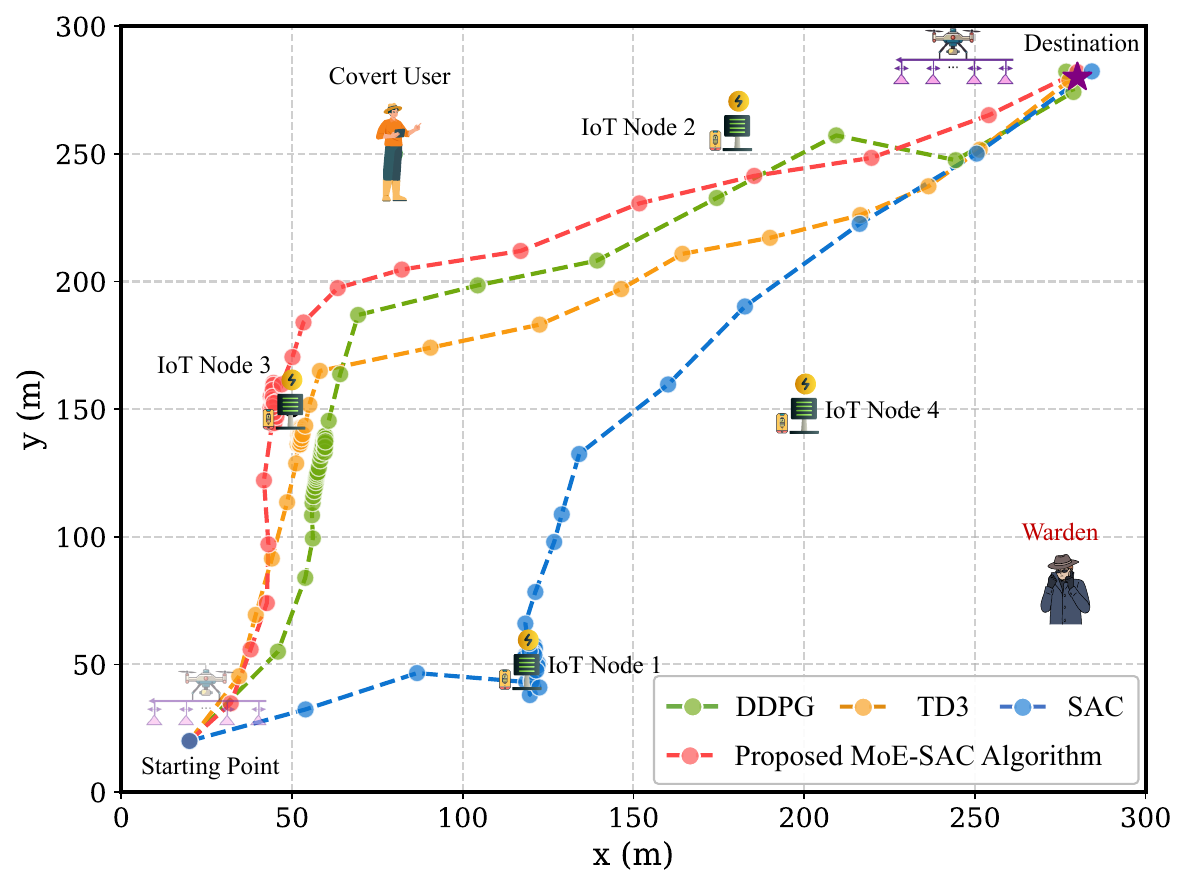}
	\caption{Trajectory of the low-altitude UAV under the different algorithms.}
	\label{Fig: Trajectory}
\end{figure}

\section{Conclusion}
\label{sec_conclusion}

\par In this paper, we have designed a novel low-altitude UAV-carried MA-enhanced transmission system for joint WPT and covert communications. Then, we have derived the covert requirement and formulated a multi-objective optimization problem aiming to maximize the total harvested energy of all IoT nodes and sum achievable rate of the covert user while minimizing the propulsion energy consumption of the low-altitude UAV. To solve this non-convex and temporally coupled optimization problem, we have reformulated it as an MDP and proposed the MoE-SAC algorithm that combines a sparse Top-K gated mixture-of-shallow-experts architecture with an action projection module to represent multimodal policy distributions and enforce critical system constraints. Simulation results have demonstrated that the MoE-SAC algorithm significantly outperforms several baseline approaches and state-of-the-art DRL algorithms, achieving superior performance in terms of the total harvested energy of all IoT nodes, sum achievable rate of the covert user, and propulsion energy consumption of the low-altitude UAV.

\ifCLASSOPTIONcaptionsoff
\newpage
\fi

\bibliographystyle{IEEEtran}
\bibliography{references.bib}

\vfill

\end{document}